\documentclass[12pt]{amsart}
\usepackage{tikz-cd}
\usepackage{amssymb}
\usepackage{amsmath}
\allowdisplaybreaks[1]
\usepackage{graphicx,color}
\usepackage{wrapfig,framed}
\usepackage[height=23cm, width=16.0cm, hmarginratio={1:1}]{geometry}
\usepackage{hyperref}
\usetikzlibrary{decorations.markings}
\tikzset{double line with arrow/.style args={#1,#2}{decorate,decoration={markings,%
mark=at position 0 with {\coordinate (ta-base-1) at (0,1pt);
\coordinate (ta-base-2) at (0,-1pt);},
mark=at position 1 with {\draw[#1] (ta-base-1) -- (0,1pt);
\draw[#2] (ta-base-2) -- (0,-1pt);
}}}}

\theoremstyle{plain}
\newtheorem{theorem}{Theorem}
\newtheorem{prop}{Proposition}
\newtheorem{lemma}{Lemma}
\newtheorem{cor}{Corollary}
\theoremstyle{definition}

\newcommand{\beq}{\begin{equation}}
\newcommand{\eeq}{\end{equation}}
\newcommand{\nn}{\nonumber}
\newcommand{\F}{\mathcal{F}}

\newcommand{\PP}{{\mathbb P}}

\newcommand{\RR}{{\mathbb R}}
\newcommand{\CC}{{\mathbb C}} 
\newcommand{\QQ}{{\mathbb Q}}

\newcommand{\ZZ}{{\mathbb Z}}

\newcommand{\tr}{{\rm tr}}
\newcommand{\e}{\epsilon}
\newcommand{\p}{\partial}

\begin{document}

\title[GUE via Frobenius manifolds]
{GUE via Frobenius Manifolds. I. From Matrix Gravity 
to Topological Gravity and Back}
\author{Di Yang}
\address{School of Mathematical Sciences, University of Science and Technology of China, Hefei 230026, P.R.~China}
\email{diyang@ustc.edu.cn}
\date{}
\begin{abstract}
Dubrovin establishes the relationship between 
the GUE partition function and the partition function 
of Gromov--Witten invariants of the complex projective line. 
In this paper, we give a direct proof of Dubrovin's result. 
We also present in a diagram the recent progress on topological gravity and matrix gravity.
\end{abstract}

\maketitle

\tableofcontents

\section{Introduction}\label{section1}
For $n\geq1$ being an integer, denote by ${\mathcal H}(n)$ the space of $n \times n$ hermitian matrices. 
The {\it normalized Gaussian Unitary Ensemble (GUE) partition function of size~$n$} is defined by 
\beq\label{partnormaldef20220201}
Z_n^{\rm GUE1}({\bf s};\e) = 2^{-n/2} (\pi\epsilon)^{-n^2/2} \int_{{\mathcal H}(n)} e^{-{\frac1\e}  \tr \, Q(M; {\bf s})} dM,
\eeq
where ${\bf s}:=(s_1,s_2,\cdots)$, $Q(y; {\bf s})$ is a power series in~$y$ of the form
\beq\label{pot}
Q(y; {\bf s}) = \frac12 y^2 -\sum_{j\geq 1} s_{j} y^{j},
\eeq
and $dM = \prod_{1\leq i\leq n} d M_{ii} \prod_{1\leq i<j\leq n} d{\rm Re} M_{ij} d{\rm Im}M_{ij}$.
For the interest of the present paper, we understand this integral in the way of first expanding
the integrand as a power series of~${\bf s}$ then 
integrating the coefficient of each monomial of~${\bf s}$ with respect to the measure~$dM$, 
and we note that the factor $2^{-n/2} (\pi\epsilon)^{-n^2/2}$ 
 in front of this integral is a normalization factor so that $Z_n^{\rm GUE1}({\bf 0};\e)\equiv1$. 

The integral in~\eqref{partnormaldef20220201} is closely related
to the enumeration of ribbon graphs (cf.~\cite{BD, BIZ, BIPZ, DFGZ, EMP, HZ, thooft1, thooft2, IZ, KKN, mehta, Witten, Zhou2}). 
Denote by 
$\mathcal{R}_{f; j_1,\dots, j_k}$ the set of oriented not-necessarily connected  
ribbon graphs having $f$~faces and $k$~vertices with valencies $j_1$, \dots, $j_k$, 
and by $\mathcal{R}^{\rm conn}_{f;j_1,\dots,j_k} \subset \mathcal{R}_{f;j_1,\dots,j_k}$ the subset consisting of the {\it connected} ones.
Then the partition function 
$Z_n^{\rm GUE1}({\bf s};\e)$ has the expression:
\begin{align}\label{ZnGUE1expan}
Z_n^{\rm GUE1}({\bf s};\e) =  1+
 \sum_{f, k\geq 1} \sum_{j_1, \dots, j_k\geq 1}  b(f;{\bf j})  s_{j_1} \cdots s_{j_k} n^f
 \e^{\frac{|{\bf j}|}2-k}, 
\end{align}
where ${\bf j}=(j_1,\dots,j_k)$, $|{\bf j}|=j_1+\cdots+j_k$, and
\beq
b(f;{\bf j}) = \sum_{G\in \mathcal{R}_{f;{\bf j}}} \frac{j_1 \cdots j_k}{|{\rm Aut} (G)|}.
\eeq
Applying Euler's formula to ribbon graphs, we see that 
$Z_n^{\rm GUE1}({\bf s};\e) \in \QQ\bigl[n,\e,\e^{-1}\bigr][[{\bf s}]]$.  
By further taking the logarithms on both sides of~\eqref{ZnGUE1expan} we obtain 
\beq
\log Z_n^{\rm GUE1}({\bf s};\e) =  
 \sum_{f,k\geq 1} \sum_{j_1, \dots, j_k\geq 1} a(f;{\bf j})  s_{j_1} \cdots s_{j_k} n^{f}
 \e^{\frac{|{\bf j}|}2-k} \in \QQ\bigl[n,\e,\e^{-1}\bigr][[{\bf s}]],
\eeq
where 
\beq
a(f;{\bf j}) = \sum_{G\in \mathcal{R}^{\rm conn}_{f;{\bf j}}} \frac{j_1 \cdots j_k}{|{\rm Aut} (G)|}.
\eeq

Following t'Hooft~\cite{thooft1, thooft2}, introduce 
\begin{equation}\label{thooft}
x := n \epsilon.
\end{equation}
Define
\beq
Z^{\rm GUE1}(x,{\bf s};\e)=Z_{x/\e}^{\rm GUE1}({\bf s};\e), \quad 
\mathcal{F}^{\rm GUE1}(x,{\bf s};\e)= \log Z_{x/\e}^{\rm GUE1}({\bf s};\e),
\eeq
and we have
\begin{align}
& Z^{\rm GUE1}(x,{\bf s};\e) =  1+
 \sum_{k,f\geq 1} \sum_{j_1, \dots, j_k\geq 1} b(f;{\bf j})  s_{j_1} \cdots s_{j_k} x^f
 \e^{\frac{|{\bf j}|}2-k-f} \in \QQ[x]\bigl(\e^2\bigr)[[{\bf s}]], \label{Zgue1x} \\
& \F^{\rm GUE1}(x,{\bf s};\e) = 
 \sum_{k\geq 1} \sum_{g\geq 0, \, j_1, \dots, j_k\geq 1\atop 2-2g - k + |{\bf j}|/2\ge1} a_g({\bf j}) 
 s_{j_1} \cdots s_{j_k} \e^{2g-2} x^{2-2g - k + \frac{|{\bf j}|}2} \in 
 \e^{-2}\QQ\bigl[x,\e^2\bigr][[{\bf s}]],
 \label{Fgue1x}
\end{align}
where $a_g({\bf j}):=a(2-2g - k+|{\bf j}|/2; {\bf j})$, 
and we used Euler's formula 
\beq
k-\frac{|{\bf j}|}2+f = 2-2g
\eeq
for a connected ribbon graph of genus~$g$. 
We call $Z^{\rm GUE1}(x,{\bf s};\e)$ the {\it normalized GUE partition function}, and 
$\mathcal{F}^{\rm GUE1}(x,{\bf s};\e)$ the {\it normalized GUE free energy}.

As in e.g.~\cite{Du09, DY1}, define the {\it corrected GUE free energy} $\F(x,{\bf s};\e)$ by 
\beq\label{defcorrGUEfree}
\mathcal{F}(x, {\bf s}; \epsilon) 
:= \F^{\rm GUE1}(x,{\bf s};\e) + \frac{x^2}{2\e^2}\biggl(\log x-\frac32\biggr) 
- \frac{\log x}{12} + \zeta'(-1) + \sum_{g\geq2} \frac{\e^{2g-2} B_{2g}}{4g(g-1)x^{2g-2}},
\eeq
where $\zeta(s)$ denotes the Riemann zeta function, and $B_m$ denote the $m$th Bernoulli number.
The {\it corrected GUE partition function} $Z(x,{\bf s};\e)$ is defined as $e^{\F(x,{\bf s};\e)}$. 
Clearly, 
\beq
\e^2 \mathcal{F}(x, {\bf s}; \epsilon) \in \QQ[[\e^2]][[x-1,{\bf s}]], \quad Z(x,{\bf s};\e) \in \QQ((\e^2))[[x-1,{\bf s}]].
\eeq
Equivalently, the corrected GUE partition function $Z(x,{\bf s};\e)$ can be defined as
\beq\label{globaldefZ}
\frac{(2\pi)^{-n}\e^{-\frac1{12}}}{{\rm Vol}(n)} \int_{{\mathcal H}(n)} e^{-{\frac1\e}  \tr \, Q(M; {\bf s})} dM, \quad x=n\e,
\eeq
where 
\beq\label{voldef41}
{\rm Vol}(n) := {\rm Vol}\left(U(n)/U(1)^n\right) = \frac{\pi^{\frac{n(n-1)}2}}{G(n+1)}, \quad G(n+1)=\prod_{j=1}^{n-1} j! .
\eeq
To see this equivalence, we view $G(n+1)$ as an analytic function ($G$ denotes Barnes' $G$-function), then  
together with~\eqref{Zgue1x} we see that the coefficient of each monomial of ${\bf s}$ in $Z(x, {\bf s};\e)$ defined from~\eqref{globaldefZ} 
is an analytic function of $x,\e$, and by taking the $n=x/\e \to \infty$ asymptotics in these coefficients 
we obtain the equivalence. Here one needs to use the fact that  
Barnes' $G$-function (cf.~\cite{Barnes,FL,WW}) has the asymptotic expansion:
\beq
\log G(z+1)\sim \frac{z^2}2 \biggl(\log z -\frac32\biggr) + \frac{z}2\log{2\pi} - \frac{\log z}{12} + \zeta'(-1) + \sum_{\ell\geq 1} \frac{B_{2\ell+2}}{4 \ell (\ell+1)z^{2\ell}}.
\eeq
For simplicity of terminology, we refer to the corrected GUE partition function (resp. corrected GUE free energy) as 
the {\it GUE partition function} (resp. {\it GUE free energy}), 
as we do in e.g.~\cite{DLYZ2, DY1, DY2}.  
Let $\mathcal{F}_g(x,{\bf s}):={\rm Coef}(\e^{2g-2},\mathcal{F}(x, {\bf s}; \epsilon))$, $g\geq0$. We call
$\mathcal{F}_g(x,{\bf s})$ the {\it genus~$g$ part of the GUE free energy} (for short the {\it genus~$g$ GUE free energy}).
It is also helpful to recall that the GUE partition function $Z(x,{\bf s};\e)$ 
satisfies the following dilaton and string 
equations, respectively:
\begin{align}
& \sum_{j\geq1} \Bigl(s_j-\frac12\delta_{j,2}\Bigr) \frac{\p Z(x,{\bf s};\e)}{\p s_{j}} + x \frac{\p Z(x,{\bf s};\e)}{\p x} +\e \frac{\p Z(x,{\bf s};\e)}{\p \e} + \frac{Z(x,{\bf s};\e)}{12}  = 0, \label{dilaton} \\
& \sum_{j\geq1} j \Bigl(s_j-\frac12\delta_{j,2}\Bigr) \frac{\p Z(x,{\bf s};\e)}{\p s_{j-1}} + \frac{xs_1}{\e^2} Z(x,{\bf s};\e) = 0. \label{string} 
\end{align}

The geometric way in understanding the GUE partition function 
is through the theory of integrable systems (Frobenius manifolds, tau-functions, bihamiltonian structures, etc.). Denote by 
$\Lambda=e^{\e \p_x}$ the shift operator, and let 
\beq\label{defVWintro}
V(x,{\bf s};\e) = \e (\Lambda-1) \frac{\p \log Z(x,{\bf s};\e)}{\p s_1}, \quad W(x,{\bf s};\e)=\e^2\frac{\p^2 \log Z(x,{\bf s};\e)}{\p s_1 \p s_1}.
\eeq
It is known (cf.~e.g.~\cite{AvM}) that the power series 
$(V(x,{\bf s}; \e), W(x,{\bf s};\e))$ is a particular solution to the 
 {\it Toda lattice hierarchy}~\cite{F, mana}, i.e., the difference 
 operator~$L$ defined by 
\beq\label{laxgue}
L = \Lambda + V(x, {\bf s};\e) + W(x,{\bf s};\e) \Lambda^{-1} 
\eeq
satisfies the following Lax-type equations:
\beq\label{todahier}
\e\frac{\p L}{\p s_j} = \left[(L^j)_+,L\right], \quad j\geq1.
\eeq
Moreover, $Z(x,{\bf s};\e)$ is a tau-function of the solution $(V(x,{\bf s}; \e), W(x,{\bf s};\e))$
to the Toda lattice hierarchy. Here and below, 
for a difference operator~$P$ in its normal form $P=\sum_{m\in \ZZ} P_m \Lambda^m$, 
$P_+:=\sum_{m\geq0} P_m \Lambda^m$, $P_-:=\sum_{m<0} P_m \Lambda^m$, 
and ${\rm res} \, P:=P_0$. We note that the functions 
$V(x,{\bf s}; \e), W(x,{\bf s};\e)$ can be uniquely determined by 
the Toda lattice hierarchy along with the initial data
\beq\label{inigue}
V(x, {\bf 0}; \e) = 0, \quad W(x,{\bf 0};\e)=x.
\eeq
The definition of a tau-function for the Toda lattice hierarchy as well as the proof of these statements 
will be reviewed in Section~\ref{section2}. 
Equations~\eqref{dilaton}--\eqref{string}, the condition $\F(x,{\bf s};\e)\in \e^{-2} \QQ[[\e^2]][[x-1,{\bf s}]]$ 
and the tau-function statement all together uniquely determine the partition function $Z(x,{\bf s};\e)=e^{\F(x,{\bf s};\e)}$ up to a pure constant factor. 

The Frobenius manifold~\cite{Du93,Du96,DZ-norm,DZ1} 
that corresponds to the Toda lattice hierarchy has the potential~\cite{Du96}:
\beq\label{FP1}
F = \frac12 v^2 u + e^u,
\eeq
where $v,u$ are the flat coordinates with $\p/\p v$ being the unit vector field. More precisely, the differential of the generating function 
for the hamiltonian densities of the dispersionless Toda lattice hierarchy is a flat section of the Dubrovin 
connection~\cite{Du96} of the Frobenius manifold~\eqref{FP1}. It is helpful to note that this Frobenius manifold 
can also be obtained from the Gromov--Witten (GW) invariants of~$\mathbb{P}^1$~\cite{Du96, KM}. Indeed,
the potential~$F$ equals, up to a quadratic function, the so-called genus zero primary free energy of 
these GW invariants. 
We often call the Frobenius manifold with potential~\eqref{FP1} the $\mathbb{P}^1$-Frobenius manifold.

For a Frobenius manifold, Dubrovin~\cite{Du96} constructs an integrable 
hierarchy of tau-symmetric hamiltonian PDEs of hydrodynamic type, called the {\it principal hierarchy}. 
This integrable hierarchy has a particular solution called the {\it topological solution}. In~\cite{DZ0} Dubrovin and Zhang prove that the 
 tau-function of the topological solution to the principal hierarchy
 (exponential of the genus zero free energy for the topological solution)
  satisfies the genus zero Virasoro constraints (see also~Liu and Tian~\cite{LT}).

For a {\it semisimple} Frobenius manifold, by solving Virasoro constraints in the form of the so-called 
{\it loop equation}, Dubrovin and Zhang~\cite{DZ-norm} (cf.~\cite{Du14})  
construct the partition function of the Frobenius manifold, and use it to define the topological deformation of the principal hierarchy, 
now called the {\it Dubrovin--Zhang (DZ) integrable hierarchy} 
for the Frobenius manifold (cf.~\cite{BPS}). By their construction, the partition function of the Frobenius manifold 
is a particular tau-function called the {\it topological tau-function} for the DZ hierarchy, that 
is the tau-function of a particular solution, called the {\it topological solution}, to the DZ hierarchy. In particular, if  
the semisimple Frobenius manifold comes from the quantum cohomology of a smooth projective variety~$X$, the partition function 
of the Frobenius manifold equals the partition function of the GW 
invariants of~$X$ \cite{Du14, DZ-norm, Givental1, Givental2, Teleman}.
A key notion in the construction of Dubrovin and Zhang is the {\it jet space}~\cite{DW, DZ-norm, Getzler2, Witten},
not only because the solution (the free energy in higher genera) to the 
 DZ loop equation is represented in terms of jet variables (jets for short) leading to uniqueness, but also 
 due to the validity at the level of integrable hierarchy (free property of jets). More precisely, firstly,
the free energy in higher genera gives rise to the topological tau-function when the 
 jet variables are subjected to the topological solution of the principal hierarchy. Secondly, by construction
 the DZ hierarchy is quasi-trivial, namely, it is obtained from the principal hierarchy under a quasi-Miura transformation, which is given by the higher genera free energy in terms of jets; and 
 this is interesting, because the quasi-Miura transformation
 could be substituted by any monotone solution to the principal hierarchy and makes it become a 
 solution to the DZ hierarchy \cite{Du09, Du14, DY3, DZ-norm, DZ1, YZ}. In~\cite{YZ}, it is suggested to interpret this as 
 a universality class of Dubrovin \cite{Du06, Du08, Du10, DGKM}. 
 A particular dense subset of monotone solutions to the principal hierarchy
 can be obtained by performing time shifts in the topological solution, but there are also other interesting 
 monotone solutions. They lead to solutions to the DZ hierarchy: all these solutions are beautifully connected to the 
 topological solution (GW invariants in the case of quantum cohomology).
  
The quantum cohomology of~$\mathbb{P}^1$ gives a semisimple Frobenius manifold. As we have mentioned above, the potential
of this Frobenius manifold is given by~\eqref{FP1}, and 
the dispersionless limit of the Toda lattice hierarchy~\eqref{todahier} form a part of the principal hierarchy (often called the stationary flows) 
of this Frobenius manifold. According to Dubrovin and Zhang~\cite{DZ1} 
the corresponding DZ hierarchy is {\it normal Miura equivalent} 
to the extended Toda lattice hierarchy~\cite{CDZ}, with an explicit formula of the normal Miura transformation. 
The DZ hierarchy is quasi-trivial, so is the extended Toda lattice hierarchy. According to the above 
discussion, their quasi-trivial transformations can be obtained from 
the free energy of GW invariants of~$\mathbb{P}^1$. For example, in~\cite{YZ}, the {\it dessins/LUE solution} to the 
 Toda lattice hierarchy is considered and it is shown that this solution can be obtained by the application of quasi-triviality. 

 
In this paper we study the {\it GUE solution} to the Toda lattice hierarchy (cf. \cite{AvM, Deift, DY1, Y}) by using the DZ approach. 
Denote by 
${\bf v}(x,{\bf s})=(v(x,{\bf s}),u(x,{\bf s}))$ the unique power-series-in-${\bf s}$ solution to the 
principal hierarchy~\eqref{phdefp1} 
 with the initial condition
\beq\label{ini2022}
v(x,{\bf 0})=0, \quad u(x,{\bf 0})=x.
\eeq 
Based on the DZ approach~\cite{DZ-norm, DZ1}, we 
will give a new proof to the following theorem.
\begin{theorem} [Dubrovin~\cite{Du09}]\label{thm1}
The genus zero GUE free energy $\mathcal{F}_0(x,{\bf s})$ has the expression:
\begin{align}
\mathcal{F}_0(x,{\bf s}) = \, &
\frac12\sum_{p,q\geq0} (p+1)! (q+1)! \Bigl(s_{p+1}-\frac12\delta_{p,1}\Bigr) \Bigl(s_{q+1}-\frac12\delta_{q,1}\Bigr) 
\Omega_{2,p;2,q}({\bf v}(x,{\bf s}))  \label{F0dub} \\
& + x\sum_{p\geq0} (p+1)! \Bigl(s_{p+1}-\frac12\delta_{p,1}\Bigr) \theta_{2,p}({\bf v}(x,{\bf s}))  + \frac12 x^2 u(x,{\bf s}).  \nn
\end{align}
For $g\geq1$, the genus~$g$ GUE free energy $\mathcal{F}_g(x,{\bf s})$ satisfies that
\beq\label{fgfmgequal}
\F_g(x,{\bf s}) = F_g^{\mathbb{P}^1}\biggl({\bf v}(x,{\bf s}), \frac{\p {\bf v}(x,{\bf s})}{\p x},\dots,\frac{\p^{3g-2} {\bf v}(x,{\bf s})}{\p x^{3g-2}}\biggr) 
+ \Bigl(\zeta'(-1)-\frac1{24} \log(-1)\Bigr)\delta_{g,1}.
\eeq
Here, $F_g^{\mathbb{P}^1}(z_1,\dots,z_{3g-2})$ $(g\geq1)$ denotes the genus~$g$ free energy in jets of the 
$\mathbb{P}^1$-Frobenius manifold
(see~\eqref{fgfmgequalp1} of Section~\ref{section3}).
\end{theorem}
Originally, a proof of this theorem was outlined in~\cite{Du09}, 
where the terminology of {\it vacuum tau-function}~\cite{DZ-norm} is used. 
Dubrovin also found~\cite{Du09}  
that the GUE partition function can be identified with part of the partition function of the Frobenius manifold with potential 
\beq\label{FNLS}
F = \frac12 (u^1)^2 u^2 + \frac12 (u^1)^2 \log u^1 -\frac34 (u^1)^2,
\eeq
leading to a second proof of the theorem (cf.~\cite{CDZ, FY}) via space/time 
duality (in genus zero: the Legendre-type transformation~\cite{Du96} of Dubrovin). 
This Frobenius manifold is often called an {\it NLS Frobenius manifold} \cite{CDZ, CvdLPS, Du96}, and will be 
discussed in details in the next of the article-series. 
Our proof given in Section~\ref{section3} will be a relatively more direct one, 
which is similar to the one given recently in~\cite{YZ} for a result for 
Grothendieck's dessins d'enfant/Laguerre Unitary Ensemble (LUE).

According to Witten~\cite{Witten}, the GUE partition function being restricted to the even couplings 
corresponds to the matrix gravity. We find that performing a further restriction given by a certain explicit and rigorous limit in the 
jet space for the higher genera parts for the even GUE partition function yields 
those for Witten's topological gravity (the celebrated Witten--Kontsevich partition function); see Corollary~\ref{cor1}.  
That means that, at least for these higher genera parts, 
the matrix gravity contains all the information of the topological gravity. 
Usually, to come back to the matrix gravity, one needs a deformation theory~\cite{DLYZ1, DLYZ2, FP, Givental2}. 
But, the recent studies~\cite{DLYZ2, DY2, DY4, YZQS} all together show that one can start with Witten's topological gravity 
and come back to the matrix gravity {\it without} a deformation theory, again at least in the higher genera.  
More precisely, we first go to the special cubic Hodge partition function by a {\it space/time duality}~\cite{Du96, YZa, YZQS} (see also~\cite{A, A2}) 
(in~\cite{YZQS} this is revealed by the {\it Hodge-BGW correspondence}), 
and then go to the so-called modified even GUE partition function by the Hodge-GUE correspondence~\cite{DLYZ2, DY2}, 
and finally back to the even GUE partition function via a product formula~\cite{DY4}, again at least to the higher genera in jets.   
(We note that the genus zero parts for the above-mentioned models are relatively easy, so for us the non-trivial things 
are in higher genera.) As a summary, we draw the following diagram:
\begin{center}
\begin{tikzcd}
 \quad F^{\rm WK}_g \arrow[leftrightarrow, "via\,s./t.\,duality"]{d} \quad  
 & \quad F^{\rm even}_g \arrow[leftrightarrow,"via\,a\,pro.\,formula"]{d}  \arrow[l] \quad 
 & \quad F_g^{\mathbb{P}^1} \arrow[l] \\
 \quad H_g \arrow[bend left]{u} \arrow[-,double line with arrow={-,-}, "{\rm Hodge-GUE}"]{r} \quad
 & \quad \widetilde F_g \quad 
 &  
\end{tikzcd}
\end{center}
Here, $g\geq1$, each of these functions lives in a certain jet space, and $F_g^{\mathbb{P}^1}$, 
$F^{\rm WK}_g$, $F^{\rm even}_g$, $\widetilde F_g$, $H_g$ stand for the genus~$g$ free energies in jets 
for GW invariants of~$\PP^1$, the Witten--Kontsevich correlators, the even GUE correlators, 
 the modified even GUE correlators, and certain special cubic Hodge integrals, respectively. 
Each one-direction arrow means taking a certain 
restriction or say limit (see \eqref{fevengdef}, \eqref{146428}--\eqref{147428}, \eqref{lowestdegreewk}), 
the long ``$=$" simply means equal (see~\eqref{jetFHequal410}), the double-direction arrow between $F^{\rm WK}_g$ and $H_g$ 
means the two are related by an invertible change of their independent jet-variables (see~\eqref{hgfwkg425} or~\eqref{123428}) up to a scalar $4^{g-1}$, 
and the double-direction arrow between $F^{\rm even}_g$'s  and $\widetilde F_g$'s 
means they are related by an invertible operation (see~\eqref{137428} or~\eqref{138428}; note that 
in this case there is a shuffling in genus), which comes from an invertible 
{\it product formula}~\cite{DY4}.

\medskip

\noindent {\bf Organization of the paper}
In Section~\ref{section2}, we review Toda lattice hierarchy 
and GUE. In Section~\ref{section3}, we prove Theorem~\ref{thm1}. 
In Section~\ref{section4} we present a discussion on topological gravity and matrix gravity.

\smallskip

\noindent {\bf Acknowledgements} 
The author is grateful to Boris Dubrovin, Don Zagier, Youjin Zhang and Jian Zhou for their advice. The work is  
partially supported by NSFC No.~12061131014.

\section{Frobenius manifold, Toda lattice hierarchy and GUE}\label{section2}

This section contains materials of several known results about $\mathbb{P}^1$-Frobenius manifold, 
Toda lattice hierarchy and GUE. We refer~\cite{Du96, DY1, Y} to the reader for further interest.

\subsection{Principal hierarchy and genus zero free energy}
Consider the $\mathbb{P}^1$-Frobenius manifold, denoted by~$M$, which has the potential~\eqref{FP1}. 
Denote by $\eta$ the invariant flat metric, and denote $v^1=v$, $v^2=u$, ${\bf v}=(v,u)$.
Following~\cite{DZ-norm, DZ1} (see also~\cite{DLYZ1}), we fix the calibration $\theta_{\alpha,p}({\bf v})$ ($\alpha=1,2, p\geq0$) for this Frobenius manifold  
via the generating series
\begin{align}
&\theta_1({\bf v};z):=\sum_{p\geq0} \theta_{1,p}({\bf v}) z^p 
= -2e^{zv}\sum_{m\geq0}\Bigl(\gamma-\frac12 u+\psi(m+1)\Bigr)e^{m u} \frac{z^{2m}}{m!^2}, \label{theta1z}\\
&\theta_2({\bf v};z):=\sum_{p\geq0} \theta_{2,p}({\bf v}) z^p = z^{-1} \biggl(\sum_{m\geq0} e^{m u+z v} \frac{z^{2m}}{m!^2}-1\biggr), \label{theta2z}
\end{align}
where $\gamma$ is the Euler constant and $\psi$ denotes the digamma function. (We recall that the calibration is a 
choice of a family of tau-symmetric hamiltonian densities for the principal hierarchy.)
The associated {\it principal hierarchy}~\cite{Du96} reads 
\begin{align}\label{phdefp1}
\frac{\p v^\alpha}{\p T^{\beta,q}} = 
\sum_{\gamma=1}^2 \eta^{\alpha\gamma}\p_x \biggl(\frac{\p \theta_{\beta,q+1}({\bf v})}{\p v^\gamma}\biggr), \quad q\geq0, \, \alpha,\beta=1,2,
\end{align}
where $\eta^{\alpha\gamma}=\delta_{\alpha+\gamma,3}$, $\alpha,\gamma=1,2$. 
As in~\cite{Du96}, define a family of holomorphic functions $\Omega_{\alpha,p;\beta,q}^{[0]}({\bf v})$ on~$M$, called the  
 {\it genus zero two-point correlation functions}, via
\beq\label{two-pt}
\sum_{p,q\geq 0} \Omega_{\alpha,p;\beta,q}^{[0]}({\bf v}) z^p y^q= \frac{1}{z+y} 
\Biggl(\sum_{\rho,\sigma=1}^2\frac{\p \theta_\alpha({\bf v}; z)}{\p v^\rho} \eta^{\rho\sigma} \frac{\p \theta_\beta({\bf v};y)}{\p v^\sigma} 
- \eta_{\alpha\beta} \Biggr),\quad \alpha,\beta=1,2.
\eeq
For an arbitrary solution ${\bf v}({\bf T})$ to the principal hierarchy~\eqref{phdefp1}, there exists a function $\F_0^M({\bf T})$ such that
\beq\label{deftaustructureg0}
\frac{\p^2 \F_0^M({\bf T})}{\p T^{\alpha,p}\p T^{\beta,q}} = \Omega_{\alpha,p;\beta,q}^{[0]}({\bf v}({\bf T})), \quad \alpha=1,2, \, p,q\geq0.
\eeq 
We call $\F_0^M({\bf T})$ the {\it genus zero free energy} of the solution ${\bf v}({\bf T})$ to the principal hierarchy~\eqref{phdefp1}, 
and call the exponential $\exp(\F_0^M({\bf T}))$
 the {\it tau-function} of the solution ${\bf v}({\bf T})$ to the principal hierarchy~\eqref{phdefp1}. 
 
As we have briefly mentioned in the Introduction, 
the $T^{2,q}$-flows of the principal hierarchy~\eqref{phdefp1} coincide with the 
dispersionless limit of the Toda lattice hierarchy~\eqref{todahier} under
\beq\label{normalflow}
\frac{\p}{\p T^{2,p}}=\frac1{(p+1)!}\frac{\p}{\p s_{p+1}}, \quad p\geq0.
\eeq 
We also mentioned in the Introduction that the potential $F$ of the $\mathbb{P}^1$-Frobenius manifold 
equals, up to a quadratic function, the genus zero primary free energy of 
the GW invariants of~$\mathbb{P}^1$. More details about the GW invariants will be given in Section~\ref{section3}.

\subsection{Review on tau-functions for the Toda lattice hierarchy}
Let 
\beq
\mathcal{A} := \mathbb Z[V(x), W(x), V(x\pm \e), W(x\pm\e), \dots]
\eeq
be the ring of polynomials with integer coefficients. The second-order difference operator 
$L= \Lambda + V(x) + W(x) \Lambda^{-1} $ (cf.~\eqref{laxgue})
can be written in the matrix form
\beq
\mathcal{L} =  \Lambda + U(\lambda) , \quad U(\lambda)= 
\begin{pmatrix} V(x)-\lambda & W(x) \\ -1 & 0 \end{pmatrix},
\eeq
where we recall that $\Lambda$ is the shift operator: $\Lambda=e^{\e \p_x}$.
\begin{lemma}[\cite{DY1}] \label{lemmaone} There exists a unique $2\times 2$ matrix series 
\beq
R(\lambda)=
\begin{pmatrix} 1 & 0\\0 & 0\end{pmatrix}
+{\rm O}\bigl(\lambda^{-1}\bigr) \in {\rm Mat} \bigl(2, \mathcal{A}[[\lambda^{-1}]]\bigr)
\eeq
satisfying equation
\beq\label{eqres}
\Lambda(R(\lambda)) U(\lambda)-U(\lambda) R(\lambda)=0
\eeq
along with the normalization conditions
\beq\label{normres}
\tr \,R(\lambda)=1, \quad \det R(\lambda)=0.
\eeq
\end{lemma}

The unique series $R(\lambda)$ in the above lemma is called the {\it basic matrix resolvent} of~$\mathcal{L}$. 
Following~\cite{DY1, Y}, define $\omega_{i,j} \in \mathcal{A}$ ($i,j\geq1$) via the generating series 
\begin{align}
&
\sum_{i,j\geq 1} \frac{\omega_{i,j}}{\lambda^{i+1} \mu^{j+1}}
=\frac{{\rm tr}\, R(\lambda)R(\mu)-1}{(\lambda-\mu)^2},
\label{taunstr1}
\end{align}
and define $\varphi_j={\rm Coef} \bigl(\lambda^{-j-1},(\Lambda(R(\lambda)))_{21}\bigr)\in \mathcal{A}$, $j\geq1$.

\begin{lemma}[\cite{DY1}]\label{lemmatwo} For an arbitrary solution $(V(x,{\bf s};\e), W(x,{\bf s};\e))$ to the Toda lattice hierarchy~\eqref{todahier}, 
there exists a function $\tau(x,{\bf s};\e)$ such that
\begin{align}
&
\e^2\frac{\partial^2\log \tau(x,{\bf s};\e)}{\partial s_{i} \partial s_{j}}
= \omega_{i,j}(x,{\bf s};\e), \quad i,j\geq1,
\label{taun1}\\
& \frac{\p}{\p s_{j}} \log \frac{\tau(x+\e,{\bf s};\e)}{\tau(x,{\bf s};\e)} = \varphi_j(x,{\bf s};\e), \quad j\geq1,
\label{taun2} \\
&
\frac{\tau(x+\e,{\bf s};\e) \tau(x-\e,{\bf s};\e)}{\tau(x,{\bf s};\e)^2}=W(x,{\bf s};\e),
\label{taun3}
\end{align}
where $\omega_{i,j}(x,{\bf s};\e)$ and $\varphi_j(x,{\bf s};\e)$ mean the substitution of $(V(x,{\bf s};\e), W(x,{\bf s};\e))$ 
in the corresponding elements in~$\mathcal{A}$. 
\end{lemma}
The function $\tau(x,{\bf s};\e)$ is determined by the solution $(V(x,{\bf s};\e), W(x,{\bf s};\e))$ up to 
\beq
\tau(x,{\bf s};\e) \mapsto e^{c + b n + \sum_{j\geq 1} a_{j-1} s_j} \tau(x,{\bf s};\e),
\eeq
where $c,b$ and $a$'s can depend on~$\e$.
We call the function $\tau(x,{\bf s};\e)$ 
the {\it DZ type tau-function} of the solution $(V(x,{\bf s};\e), W(x,{\bf s};\e))$ 
to the Toda lattice hierarchy, for short a Toda tau-function of the solution $(V(x,{\bf s};\e), W(x,{\bf s};\e))$. The 
elements $\omega_{i,j}\in\mathcal{A}$ ($i,j\geq1$) are called two-point correlation functions
of the Toda lattice hierarchy.

Our definition of a Toda tau-function agrees with 
the one given in~\cite{CDZ, DZ1}. Indeed, 
firstly, recall from~\cite{DY1} that the two-point correlation functions
$\omega_{i,j}$ are associated with the 
tau-symmetric hamiltonian densities for the Toda lattice hierarchy; secondly, 
 by taking the dispersionless 
limit $\e\to 0$ in the above definition (symbolically put $v=\lim_{\e\to0} V(x)$ and $w=\lim_{\e\to 0} W(x)$) and by using Lemma~\ref{lemmaone} (see~\cite{DY1, Y} for some more details), 
one immediately obtains that  
\begin{align}
& \sum_{i,j\geq 1} \frac{\omega_{i,j}^{[0]}}{\lambda^{i+1} \mu^{j+1}}
=\frac{B(\lambda) B(\mu)((\lambda-v)(\mu-v)-4w)-1}{2 \, (\lambda-\mu)^2},
\label{taun1disp}\\
&
\frac1\lambda + \sum\limits_{i\geq 1} \frac{\varphi_i^{[0]}}{\lambda^{i+1}}  =  B(\lambda),
\label{taun2disp} 
\end{align}
where 
\beq
B(\lambda) = \frac1{\sqrt{(\lambda-v)^2-4w}}
=\frac{1}{\lambda}+\frac{v}{\lambda ^2}+\frac{v^2+2 w}{\lambda ^3}+\frac{v^3+6 v w}{\lambda ^4}+{\rm O}\bigl(\lambda^{-5}\bigr),
\eeq
and the verification in the dispersionless limit, i.e. of the following equalities 
\beq
\omega_{i,j}^{[0]} = i! j! \Omega^{[0]}_{2,i-1;2,j-1}, \quad \varphi_{i}^{[0]} = i! \theta_{2,i-1}, \quad i,j\geq1,
\eeq
(cf.~\eqref{theta2z}, \eqref{two-pt}, \eqref{normalflow})
is straightforward.

%
%

\subsection{GUE partition function as a Toda tau-function}
Recalling Gaussian integral
\beq\label{Gaussian}
\int_{{\mathcal H}(n)} e^{-\frac1{2\epsilon} \tr\, M^2} dM=2^{\frac{n}2} (\pi\epsilon)^{\frac{n^2}2}, \quad n\geq1,
\eeq
we get $Z^{\rm GUE1}_n({\bf 0};\e) \equiv1$.  
For $Z^{\rm GUE1}_n({\bf s};\e)$, there is the well-known formula
(cf.~\cite{Deift, DY1, mehta})
\begin{align}
\int_{{\mathcal H}(n)} e^{-\frac1{\epsilon} \tr Q(M;{\bf s})} dM
=& \frac1{n!} {\rm Vol} (U(n)/ U(1)^n)\int_{\RR^n}  \Delta_n(\lambda_1,\dots,\lambda_n)^2 e^{-\frac1{\epsilon} 
\sum_{k=1}^n Q(\lambda_k;{\bf s})}d\lambda_1 \cdots d\lambda_n, \nn
\end{align}
where
$$
\Delta_n(\lambda_1,\dots,\lambda_n)=\prod_{1\leq i<j\leq n} (\lambda_i-\lambda_j), \quad n\geq1.
$$

One can apply the theory of orthogonal polynomials (cf.~e.g.~\cite{Deift})
for a further computation of $Z^{\rm GUE1}_n({\bf s};\e)$.
Let $(\cdot,\cdot)$ be an inner product on the space of polynomials defined by 
\beq\label{inner}
(f,g)=\int_{-\infty}^{+\infty} f(\lambda)  g(\lambda) e^{-\frac1{\epsilon} Q(\lambda; {\bf s})}d\lambda, \quad \forall\, f,g.
\eeq
Let
\beq\label{ortopol1}
p_m=p_m(\lambda; {\bf s};\e)=\lambda^m+a_{1m}({\bf s};\e) \lambda^{m-1}+\dots+a_{mm}({\bf s};\e), \quad m\geq0,
\eeq
be a system of monic polynomials orthogonal with respect to $(\cdot,\cdot)$, i.e.,
\beq\label{ortopol2}
(p_{m_1}(\lambda; {\bf s};\e), p_{m_2}(\lambda; {\bf s};\e)) =: h_{m_1}({\bf s};\e) \delta_{m_1 m_2}, \quad \forall\, m_1,m_2\geq0.
\eeq
Observing that $\Delta_n(\lambda_1,\dots,\lambda_n)$ can be written into the form
\beq
\Delta_n(\lambda_1,\dots,\lambda_n)=\det  
\left( \begin{array}{cccc} p_0(\lambda_1;{\bf s};\e) & p_0(\lambda_2;{\bf s};\e) & \dots & p_0(\lambda_n;{\bf s};\e)\\
p_1(\lambda_1;{\bf s};\e) & p_1(\lambda_2;{\bf s};\e) & \dots & p_1(\lambda_n;{\bf s};\e)\\
\cdot & \cdot & \dots & \cdot\\
\cdot & \cdot & \dots & \cdot\\
\cdot & \cdot & \dots & \cdot\\
p_{n-1}(\lambda_1;{\bf s};\e) & p_{n-1}(\lambda_2;{\bf s};\e) & \dots & p_{n-1}(\lambda_n;{\bf s};\e)
\end{array}\right)_{n\times n}, 
\eeq
we find
\beq\label{final}
\int_{{\mathcal H}(n)} e^{-\frac1{\epsilon} \tr \, Q(M;{\bf s})} dM={\rm Vol} (U(n)/U(1)^n) h_0({\bf s};\e) \cdots h_{n-1}({\bf s};\e), \quad n\geq1.
\eeq

At ${\bf s}=0$ the orthogonal polynomials $p_m(\lambda; {\bf 0};\e)$ have the explicit expressions
\beq
p_m(\lambda; {\bf 0})=\epsilon^{\frac{m}2} {\rm He}_m(\lambda/\e^{1/2}), \quad m\geq0,
\eeq
where ${\rm He}_m(s)$ are the hermite polynomials. Recalling that 
\beq
\int_{-\infty}^\infty {\rm He}_{m_1}(t) {\rm He}_{m_2}(t) e^{-\frac12 t^2} dt=\sqrt{2\pi} m_1! \delta_{m_1m_2}.
\eeq
we find 
\beq\label{hm0}
h_m({\bf s}={\bf 0};\e)= \epsilon^{m+\frac12} \sqrt{2\pi} m!.
\eeq
Using \eqref{Gaussian}, \eqref{final} and~\eqref{hm0}, we obtain 
\beq
{\rm Vol} (U(n)/U(1)^n)=\frac{\pi^{\frac{n^2-n}2}} {\prod_{j=1}^{n-1} j!} \,.
\eeq
Therefore, 
\beq\label{gue}
Z_n^{\rm GUE1}({\bf s}; \epsilon) = 
\frac{(2\pi)^{-\frac{n}2} \epsilon^{-\frac{n^2}2}}{\prod_{j=1}^{n-1} j!} h_0({\bf s};\e) \cdots h_{n-1}({\bf s};\e), \quad n\geq1.
\eeq
We also define $Z_0^{\rm GUE1}({\bf s}; \epsilon)\equiv1$.

The orthogonal polynomials $p_m(\lambda;{\bf s}; \epsilon)$ satisfy the 
the {\it three-term recurrence relation}:
\beq\label{lax}
 p_{m+1}(\lambda;{\bf s}; \epsilon) + V_m({\bf s};\e) p_m(\lambda;{\bf s}; \epsilon) + W_m({\bf s};\e) p_{m-1}(\lambda;{\bf s}; \epsilon)= \lambda p_m(\lambda;{\bf s}; \epsilon),  
 \quad m\geq0,
\eeq
for some functions $V_m({\bf s};\e)$ and $W_m({\bf s};\e)$ $(m\geq0)$, 
with $p_{-1}(\lambda;{\bf s}; \epsilon):\equiv0$ and $W_0({\bf s};\e):\equiv0$. The functions 
$V_m({\bf s};\e)$ and $W_m({\bf s};\e)$ are well known to satisfy 
\beq
V_m({\bf 0};\e)= 0 , \quad W_m({\bf 0};\e)= \e m, \quad m\geq0. 
\eeq
The equality
\beq
(\lambda p_{m_1}, p_{m_2})=(p_{m_1}, \lambda p_{m_2})
\eeq
then implies that 
\beq\label{wn}
W_m({\bf s};\e)=\frac{h_m({\bf s};\e)}{h_{m-1}({\bf s};\e)} = 
\frac{Z_{m+1}^{\rm GUE1}({\bf s}; \epsilon)Z_{m-1}^{\rm GUE1}({\bf s}; \epsilon)}{Z_{m}^{\rm GUE1}({\bf s}; \epsilon)^2} + {\rm correction}, \quad m\geq1.
\eeq

The three-term recurrence relation tells that $p_0,p_1,\dots$ are eigenvectors of the difference operator 
$L=\Lambda + V_n({\bf s};\e)+W_n({\bf s};\e) \Lambda^{-1}$ with $\Lambda:f(n) \mapsto f(n+1)$ being the shift operator. 
Denote again by~$L$ the corresponding tri-diagonal matrix, and denote  ${\bf p} = (p_0, p_1,\dots)$. 
For a square matrix $X=(X_{i,j})$, denote 
$$
X_-=(X_{i,j})_{i<j}, \quad X_+=(X_{i,j})_{i\geq j}, \quad X=X_+ + X_-.
$$

\begin{lemma}\label{Lax00} 
For arbitrary $n\geq0$, the following two statements are true: a) the  
polynomial $p_n(\lambda;{\bf s}; \epsilon)$ satisfies 
\beq\label{dt}
\epsilon\frac{\p p_n(\lambda;{\bf s}; \epsilon)}{\p s_j}  -  (A_j {\bf p})_n(\lambda;{\bf s}; \epsilon)=0, \quad A_j := - (L^{j})_-, \quad j\geq 1;
\eeq
b) we have
\beq\label{dtdn}
\e \frac{\p }{\p s_j} \log \frac{Z_{n+1}^{\rm GUE1}({\bf s}; \epsilon)}{Z_n^{\rm GUE1}({\bf s}; \epsilon)} = (L^j)_{nn}, \quad j\geq1.
\eeq
\end{lemma}

\begin{proof} Write
$$
\frac{\p p_n(\lambda;{\bf s}; \epsilon)}{\p s_j} =\sum_{m=0}^{n-1} A_{m n}^{(j)}({\bf s}; \epsilon)  p_m(\lambda;{\bf s}; \epsilon)
$$
for some coefficients $A_{mn}^{(j)}=A_{m n}^{(j)}({\bf s}; \epsilon)$. Differentiating
the orthogonality $(p_m, p_n) \equiv h_n \delta_{mn}$
with respect to~$s_j$, we find that for $m<n$ 
\beq
A_{mn}^{(j)}h_m +\frac1{\epsilon}  (\lambda^{j} p_n, p_m)=0, 
\eeq
and for $m=n$,
\beq
\frac1{\epsilon}  ( \lambda^{j} p_n, p_n) = \frac{\p h_n}{\p s_j} 
= \frac{\p }{\p s_j} \log \frac{Z_{n+1}^{\rm GUE1}({\bf s}; \epsilon)}{Z_n^{\rm GUE1}({\bf s}; \epsilon)}.
\eeq
Then by using~\eqref{lax} we obtain~\eqref{dt} and~\eqref{dtdn}.
\end{proof}

If particular, from~\eqref{dtdn} we see that for all $n\geq0$, 
\beq\label{vn}
\e \frac{\p }{\p s_1} \log \frac{Z_{n+1}^{\rm GUE1}({\bf s}; \epsilon)}{Z_n^{\rm GUE1}({\bf s}; \epsilon)} = V_n({\bf s};\e), \quad j\geq1.
\eeq

It follows from Lemma~\ref{Lax00} that the difference operator~$L$ satisfies the Toda lattice hierarchy~\eqref{todahier}. 
Indeed, it immediately follows from~\eqref{dt} the Lax equations for the square matrix~$L$; 
by using~\eqref{wn}, \eqref{vn} and the facts mentioned in the introduction that 
$Z_n^{\rm GUE1}({\bf s};\e)$ and $\log Z_n^{\rm GUE1}({\bf s};\e)$ belong to $\QQ\bigl[n,\e,\e^{-1}\bigr][[{\bf s}]]$, we know that 
$W_n({\bf s};\e), V_n({\bf s};\e)$ belong to $\QQ\bigl[n,\e,\e^{-1}\bigr][[{\bf s}]]$; 
we therefore conclude that the Lax equations~\eqref{todahier} hold for the difference operator~$L$ (namely with $n$ being viewed as 
an indeterminate), i.e., $(V_n({\bf s};\e), W_n({\bf s};\e))$ is a solution to the Toda lattice hierarchy.

We note that equalities \eqref{wn}, \eqref{vn}, \eqref{dtdn} hold true in $\QQ\bigl[n,\e,\e^{-1}\bigr][[{\bf s}]]$, 
where the right-hand side of~\eqref{dtdn} should be viewed as ${\rm res} \, L^{j}$. By
using~\eqref{dtdn} and the compatibility between \eqref{taun1}, \eqref{taun2} we have 
\beq
\sum_{i,j\geq 1} \frac{1}{\lambda^{i+1} \mu^{j+1}} (\Lambda-1) \biggl(\frac{\partial^2\log Z_{n}^{\rm GUE1}({\bf s}; \epsilon)}{\partial s_{i} \partial s_{j}}\biggr)
= (\Lambda-1) \biggl(\frac{{\rm tr}\, R_n(\lambda)R_n(\mu)-1}{(\lambda-\mu)^2}\biggr).
\eeq
Since $Z_{n}^{\rm GUE1}({\bf s}; \epsilon)\in\QQ\bigl[n,\e,\e^{-1}\bigr][[{\bf s}]]$ and since $Z_{n}^{\rm GUE1}({\bf s}; \epsilon)$ is divisible by~$n$, 
we have
\beq\label{ssZ}
\sum_{i,j\geq 1} \frac{1}{\lambda^{i+1} \mu^{j+1}} \frac{\partial^2\log Z_{n}^{\rm GUE1}({\bf s}; \epsilon)}{\partial s_{i} \partial s_{j}}
= \frac{{\rm tr}\, R_n(\lambda)R_n(\mu)-1}{(\lambda-\mu)^2}.
\eeq
In particular, $W_n({\bf s};\e) = \e^2\frac{\p^2 \log Z_{n}^{\rm GUE1}({\bf s}; \epsilon)}{\p s_1 \p s_1}$. 
Since $\log Z_n^{\rm GUE1}({\bf s}; \epsilon)$ and $\log Z(x,{\bf s};\e)$ differ by a function that only depends on $x,\e$, 
we see that $W_n({\bf s};\e)$ defined in this section coincides with $W(x, {\bf s};\e)$ defined in the Introduction. 

We see also from \eqref{wn}, \eqref{dtdn}, \eqref{ssZ} that $Z_{n}^{\rm GUE1}({\bf s}; \epsilon)$ almost satisfies the definition 
of a Toda tau-function of the solution $(V_n({\bf s};\e), W_n({\bf s};\e))$, except in~\eqref{wn} an extra term appears.  
It is easy to show (cf.~\cite{AIKZ}, \cite{WW}) that the definition for the correction GUE 
partition function (cf.~\eqref{defcorrGUEfree}, \eqref{globaldefZ}) eliminates the extra term and 
keep the other properties hold. We therefore arrive at 
 the following proposition summarizing the above.
\begin{prop}[cf.~\cite{AvM, DY1}]
The vector-valued function 
$(V(x,{\bf s}; \e), W(x,{\bf s};\e))$ defined in~\eqref{defVWintro} is the unique solution to the 
Toda lattice hierarchy~\eqref{todahier} specified by the initial condition~\eqref{inigue}, 
and the GUE partition function $Z(x,{\bf s};\e)$ is the tau-function of this solution to the Toda lattice hierarchy.
\end{prop}

\section{Proof of Theorem~\ref{thm1}}\label{section3}

The goal of this section is to prove Theorem~\ref{thm1}.  To this end, we will first need to 
recall the definition of $F^{\mathbb{P}^1}_g$ ($g\geq1$) that appear in the context of Theorem~\ref{thm1}.

Denote by~$\F^{\mathbb{P}^1}({\bf T};\e)$ the {\it free energy of GW invariants of~$\mathbb{P}^1$}, 
and by $Z^{\mathbb{P}^1}({\bf T};\e):=\exp\bigl(\F^{\mathbb{P}^1}({\bf T};\e)\bigr)$ the {\it partition function} of 
these GW invariants:
\beq\label{freeenergy}
\F^{\mathbb{P}^1}({\bf T};\e) \:=\sum_{d,k\geq 0} \frac1{k!} 
\sum_{1\leq \alpha_1,\dots,\alpha_k\leq 2, \atop i_1,\dots,i_k\geq 0}  T^{\alpha_1, i_1} \cdots T^{\alpha_k,i_k} 
\sum_{g\ge0} \e^{2g-2} \langle\tau_{i_1}(\alpha_1)\cdots\tau_{i_k}(\alpha_k)\rangle_{g,d},
\eeq
where ${\bf T}:=(T^{\alpha,i})_{\alpha=1,2,\,i\geq 0}$ and  
$\langle\tau_{i_1}(\alpha_1)\cdots\tau_{i_k}(\alpha_k)\rangle_{g,d}$ denote the genus~$g$ and degree~$d$ 
GW invariants of~$\mathbb{P}^1$ (cf.~e.g.~\cite{DZ1,OP1,OP2}). We denote by 
$\F^{\mathbb{P}^1}_g({\bf T}):={\rm Coef}\bigl( \e^{2g-2},\F^{\mathbb{P}^1}({\bf T};\e)\bigr)$ the genus~$g$ part 
of~$\F^{\mathbb{P}^1}({\bf T};\e)$, sometimes called for short {\it the genus~$g$ free energy of GW invariants of~$\mathbb{P}^1$}. 
It was conjectured by Dubrovin~\cite{Du93}, Eguchi--Yang~\cite{EY} (cf.~\cite{Getzler, Zhang}), 
and proved in~\cite{DZ1, OP1, OP2} that the functions  
\beq
V^{\mathbb{P}^1}({\bf T};\e) := \e (\Lambda-1) \frac{\p \log Z^{\mathbb{P}^1}({\bf T};\e)}{\p T^{2,0}}, \quad 
W^{\mathbb{P}^1}({\bf T};\e) :=\e^2\frac{\p^2 \log Z^{\mathbb{P}^1}({\bf T};\e)}{\p T^{2,0} \p T^{2,0}},
\eeq
satisfy the Toda lattice hierarchy~\eqref{todahier} with
\beq\label{normalflowold}
\frac{\p}{\p T^{2,p}}=\frac1{(p+1)!}\frac{\p}{\p s_{p+1}}, \quad p\geq0,
\eeq
and with $L := \Lambda + V^{\mathbb{P}^1}({\bf T};\e) + W^{\mathbb{P}^1}({\bf T};\e) \Lambda^{-1}$,
and, moreover, $Z^{\mathbb{P}^1}({\bf T};\e)$ is a tau-function of the 
solution $(V^{\mathbb{P}^1}({\bf T};\e), W^{\mathbb{P}^1}({\bf T};\e))$
to~\eqref{todahier}. Here $\Lambda=\exp(\e \p_x)$ and $x=T^{1,0}$.
The partition function $Z^{\mathbb{P}^1}({\bf T};\e)$ also 
 satisfies the following dilaton and string equations:
\begin{align}
& \sum_{\alpha=1}^2 \sum_{p\geq0} \bigl(T^{\alpha,p}-\delta^{\alpha,1}\delta^{p,1}\bigr) \frac{\p Z^{\mathbb{P}^1}({\bf T};\e)}{\p T^{\alpha,p}} 
+\e \frac{\p Z^{\mathbb{P}^1}({\bf T};\e)}{\p \e} + \frac1{12}Z^{\mathbb{P}^1}({\bf T};\e) = 0, \label{dilatonp1}\\
&  
\sum_{\alpha=1}^2 \sum_{p\geq1} \bigl(T^{\alpha,p}-\delta^{\alpha,1}\delta^{p,1}\bigr) \frac{\p Z^{\mathbb{P}^1}({\bf T};\e)}{\p T^{\alpha,p-1}} + \frac{T^{1,0}T^{2,0}}{\e^2} Z^{\mathbb{P}^1}({\bf T};\e)
=0. \label{stringp1}
\end{align}

Let ${\bf v}^{\mathbb{P}^1}({\bf T})$ be the {\it topological solution} to~\eqref{phdefp1}, 
that is the unique power series in $T^{\alpha,q}$, $\alpha=1,2$, $q>0$, satisfying~\eqref{phdefp1} and 
\beq
v^{\alpha,\mathbb{P}^1}({\bf T})\big|_{T^{\beta,q}=0, q>0, \beta=1,2} = T^{\alpha,0},\quad \alpha=1,2.
\eeq
As mentioned in the Introduction, the genus zero free energy of GW 
invariants of~$\mathbb{P}^1$ equals~\cite{Du96} that for the principal hierarchy subjected to the 
topological solution, i.e.,
\beq
\F^{\mathbb{P}^1}_0({\bf T}) = 
\frac12 \sum_{\alpha,\beta=1}^2 \sum_{p,q\geq0} \bigl(T^{\alpha,p}-\delta^{\alpha,1}\delta^{p,1}\bigr) \bigl(T^{\beta,q}-\delta^{\beta,1}\delta^{q,1}\bigr)  
\Omega_{\alpha,p;\beta,q}^{[0]}\bigl({\bf v}^{\mathbb{P}^1}({\bf T})\bigr).
\eeq
The higher genus free energies $\F^{\mathbb{P}^1}_g({\bf T})$, $g\geq1$, 
admit the {\it jet-variable representation} (cf.~\cite{DW, DZ-norm, DZ1, EYY, Getzler2}).
Namely, there exist functions 
$F^{\PP^1}_g({\bf v}_0, {\bf v}_1,\dots, {\bf v}_{3g-2})$ ($g\geq1$) 
with ${\bf v}_m=(v_m,u_m)=(v^1_m,v^2_m)$ and ${\bf v}_0={\bf v}$, 
such that 
\beq\label{fgfmgequalp1}
\F_g^{\PP^1}({\bf T}) = 
F_g^{\PP^1}\biggl({\bf v}^{\PP^1}({\bf T}), \frac{\p {\bf v}^{\PP^1}({\bf T})}{\p x},\dots,\frac{\p^{3g-2} {\bf v}^{\PP^1}({\bf T})}{\p x^{3g-2}}\biggr).
\eeq
By using Virasoro constraints, Dubrovin and Zhang~\cite{DZ-norm, DZ1} obtained the following 
loop equation: 
\begin{align}
\label{loop}
&\sum_{r\geq 0}\biggl({\p \Delta F \over \p v_r}
\Bigl({v-\lambda\over D}\Bigr)_r - 2 {\p \Delta F \over \p u_r} \Bigl({1\over D}\Bigr)_r \biggr)\\
&\quad +\sum_{r\geq 1} \sum_{k=1}^r \binom{r}{k}
\Bigl({1\over \sqrt{D}}\Bigr)_{k-1} \biggl( {\p \Delta F \over \p v_r}
\left(v\!-\!\lambda\over \sqrt{D}\right)_{r-k+1}-2 {\p \Delta F \over \p u_r} \left(1\over \sqrt{D}\right)_{r-k+1} \biggr)
\nn\\
&=D^{-3} e^{u} \left(4e^{u}+(v-\lambda)^2\right)\nn\\
&\quad - \epsilon^2 \sum_{k,l\geq0} \biggl( \frac14 S(\Delta F,v_k,v_l)
\left(v-\lambda\over \sqrt{D}\right)_{k+1}  \left(v-\lambda\over
\sqrt{D}\right)_{l+1} \nn\\
&\quad - S(\Delta F,v_k,u_l) \left(v-\lambda\over
\sqrt{D}\right)_{k+1} \left(1\over \sqrt{D}\right)_{l+1}+ S(\Delta F,u_k,u_l) \,
\left(1\over \sqrt{D}\right)_{k+1} \left(1\over \sqrt{D}\right)_{l+1}\biggr)
\nn\\
&\quad -\frac{\epsilon^2} 2 \sum_{k\geq0}  \biggl( {\p \Delta F \over \p v_k}
  \p^{k+1} \biggl(e^{u} {4 e^{u}(v-\lambda) u_1 - ((v-\lambda)^2 + 4 e^{u}) v_1 \over D^3}\biggr)  \nn\\
&\qquad\qquad\qquad +{\p \Delta F \over \p u_k} \p^{k+1}\biggl(e^{u} {4 (v-\lambda) \, v_1 - ((v-\lambda)^2 + 4 e^{u}) u_1\over D^3}\biggr)\biggr),\nn
\end{align}
where $\Delta F:= \sum_{g\geq 1} \epsilon^{2g} F_g^{\mathbb{P}^1}({\bf v}_0, {\bf v}_1,\dots, {\bf v}_{3g-2})$,
$D= (v-\lambda)^2 - 4 e^{u}$,
$S(f,a,b):=
{\p^2 f\over \p a \p b}+
{\p f \over \p a}
{\p f \over \p b}$,
and $f_r$ stands for $\p^r(f)$ with 
\beq
\p:=\sum_{\alpha=1,2} \sum_{m\geq0} v^\alpha_{m+1} \frac{\p}{\p v^{\alpha}_m}.
\eeq
It is also shown in~\cite{DZ-norm, DZ1} the solution $\Delta F$ to~\eqref{loop} is unique up to
a sequence of additive constants for $F_g^{\mathbb{P}^1}$ ($g\geq1$), that for $g\geq2$ can be fixed by
the following equation:
\begin{align}
& \sum_{\alpha=1}^2 \sum_{m=1}^{3g-2} m v^\alpha_m \frac{\p F_g^{\mathbb{P}^1}({\bf v}_0, {\bf v}_1,\dots, {\bf v}_{3g-2})}{\p v^\alpha_m} 
= (2g-2) F_g^{\mathbb{P}^1}({\bf v}_0, {\bf v}_1,\dots, {\bf v}_{3g-2}) + \frac{\delta_{g,1}}{12}, ~ g\ge1.\label{homogenfgp1}
\end{align}
Moreover, for $g\geq2$, $F_g^{\mathbb{P}^1}({\bf v}_0, {\bf v}_1,\dots, {\bf v}_{3g-2})$ 
are polynomials of ${\bf v}_{2}, \dots, {\bf v}_{3g-2}$ and have rational dependence
in~${\bf v}_{1}$ (cf.~e.g.~\cite{DZ-norm, DZ1}). 
In particular, for $g=1$, 
\begin{align}\label{f1gw}
F_1^{\mathbb{P}^1}({\bf v}_0, {\bf v}_1) = \frac1{24} \log \bigl((v_1)^2-e^{u} (u_1)^2\bigr) - \frac1{24} u.
\end{align}
These unique functions $F_g^{\mathbb{P}^1}$ ($g\geq1$) are the ones used in the context of Theorem~\ref{thm1}.

We are ready to prove Theorem~\ref{thm1}.

\begin{proof}[Proof of Theorem~\ref{thm1}.]
Start with genus zero. Let $(v(x,{\bf s}), u(x,{\bf s}))$ be the unique solution 
to the initial value problem \eqref{phdefp1}, \eqref{ini2022}.
The Riemann invariants for the principal hierarchy~\eqref{phdefp1} are given by
\beq
R_1({\bf v})=v+2e^{u/2}, \quad R_2({\bf v})=v-2e^{u/2}.
\eeq
Since $(R_i)_x$ ($i=1,2$) do not vanish at generic~$x=x_0$, 
the solution $(v(x,{\bf s}), u(x,{\bf s}))$ belongs to the class of {\it monotone solutions}. Therefore, it could be obtained by the
hodograph method~\cite{Du10,DGKM, Tsarev}, yielding the following genus zero Euler--Lagrange equation:
\beq\label{ELgue45}
x \delta_{\beta,2} + \sum_{p\geq0} (T^{2,p}-\delta_{p,1}) \frac{\p \theta_{2,p}}{\p v^\beta}({\bf v}(x,{\bf s})) = 0, \quad \beta=1,2,
\eeq
where $T^{2,p}=(p+1)!s_{p+1}$, $p\geq0$.

Following~\cite{Du96}, define $\widehat\F_0(x,{\bf s})$ as the right-hand side of~\eqref{F0dub}. 
By using the well-known properties 
\begin{align}
& \Omega_{\alpha,p;\beta,q}^{[0]}({\bf v})=\Omega_{\beta,q;\alpha,p}^{[0]}({\bf v}), \quad
\p_{t^{\gamma,s}}(\Omega_{\alpha,p;\beta,q}^{[0]}({\bf v}))=\p_{t^{\beta,q}}(\Omega_{\alpha,p;\gamma,s}^{[0]}({\bf v})),
\quad \forall\,p,q,s\geq0, \label{omegapro1}\\
& \theta_{\alpha,p}({\bf v})=\Omega^{[0]}_{\alpha,p;1,0}({\bf v}), \quad \forall\, p\geq0, \label{omegapro2} 
\end{align}
one can verify the validity of the following equalities:
\begin{align}
\frac{\p^2 \widehat{\F}_0(x,{\bf s})}{\p T^{2,p} \p T^{2,q}} & = \Omega_{2,p;2,q}^{[0]}({\bf v}(x,{\bf s})), \quad \forall \, p,q\geq0, \label{defgenus0tau1}\\
\frac{\p^2 \widehat{\F}_0(x,{\bf s})}{\p x \p x} & = u(x,{\bf s}), \label{defgenus0tau2}\\
\frac{\p^2 \widehat{\F}_0(x,{\bf s})}{\p x \p T^{2,p}} & = \Omega_{1,0;2,p}^{[0]}({\bf v}(x,{\bf s})), \quad \forall \, p\geq0.  \label{defgenus0tau3}
\end{align}
From these equalities we see that 
 $\exp\bigl(\e^{-2}\widehat{\F}_0(x,{\bf s})\bigr)$ is the tau-function for the solution $(v(x,{\bf s}), u(x,{\bf s}))$
 to the $\p_x, \p_{T^{2,q}}$-flows of the principal hierarchy~\eqref{phdefp1} (cf.~\eqref{deftaustructureg0}).

It is not difficult to verify that 
 $\widehat{\F}_0(x,{\bf s})$ also satisfies the following linear equations: 
\begin{align}
&\sum_{j\geq1} \Bigl(s_j-\frac12 \delta_{j,2}\Bigr) \frac{\p \widehat \F_0(x,{\bf s})}{\p s_j} + x \frac{\widehat\F_0(x,{\bf s})}{\p x} 
= 2 \widehat \F_0(x,{\bf s}), \label{dilatongenuszero48} \\
&\sum_{j\geq2} j \Bigl(s_j-\frac12 \delta_{j,2}\Bigr) \frac{\p \widehat{\F}_0(x,{\bf s})}{\p s_{j-1}} +  xs_1 = 0. \label{genuszerostringgue}
\end{align}
We conclude from \eqref{defgenus0tau1}--\eqref{defgenus0tau3} and~\eqref{genuszerostringgue}
that $\widehat{\F}_0(x,{\bf s})$ could differ from~$\F_0(x,{\bf s})$ only possibly by adding a function
of~$x$ (actually with at most linear dependence in~$x$).
Taking ${\bf s}={\bf 0}$ in $\widehat\F_0(x,{\bf s})$ and in $\F_0(x,{\bf s})$, we find that they both give 
\beq
\frac{x^2}{2\e^2}\biggl(\log x-\frac32\biggr).
\eeq
Hence formula~\eqref{F0dub} is proved.

Similarly as we do for the LUE case in~\cite{YZ},  we proceed with the higher genera by using {\it quasi-triviality}. 
According to~\cite{DZ-norm,DZ1}, the following quasi-trivial map  
\begin{align}
& \widehat V = \frac{\Lambda-1}{\e\p_x} (v) + (\Lambda-1) \circ \p_{t^{2,0}} \Biggl(\sum_{g\geq1} \e^{2g-1} F_g^{\PP^1}\biggl({\bf v}, \frac{\p {\bf v}}{\p x},\dots,\frac{\p^{3g-2} {\bf v}}{\p x^{3g-2}}\biggr)\Biggr), \label{qm1}\\
& \widehat W = \frac{(\Lambda+\Lambda^{-1}-2)} {\e^2 \p_x^2} (u) + \bigl(\Lambda+\Lambda^{-1}-2\bigr) \Biggl(
\sum_{g\geq1} \e^{2g-2} F_g^{\PP^1}\biggl({\bf v}, \frac{\p {\bf v}}{\p x},\dots,\frac{\p^{3g-2}{\bf v}}{\p x^{3g-2}}\biggr)\Biggr), \label{qm2}
\end{align}
transforms the principal hierarchy~\eqref{phdefp1} to the extended Toda hierarchy~\cite{CDZ,DZ1}. The quasi-Miura map~\eqref{qm1}--\eqref{qm2} 
transforms a monotone solution of the principal hierarchy~\eqref{phdefp1} to a 
solution of the extended Toda hierarchy (see the Theorem~1.1 of~\cite{DZ1}).
As we just mentioned above, the particular solution $(v(x,{\bf s}), u(x, {\bf s}))$ 
of interest to the $\p_{t^{2,p}}$-flows ($p\geq0$) in the principal hierarchy~\eqref{phdefp1} specified 
by the initial data~\eqref{ini2022} is monotone. Therefore, the function  
$(\widehat V(x,{\bf s};\e), \widehat U(x,{\bf s};\e))$ defined by
\begin{align}
&\widehat V(x,{\bf s};\e) := \widehat V|_{v_k\mapsto\partial_x^k( v(x,{\bf s};\e)), u_k\mapsto\partial_x^k(u(x,{\bf s};\e)), k\geq0}, \\
&\widehat U(x,{\bf s};\e) := 
\widehat U|_{v_k\mapsto\partial_x^k( v(x,{\bf s};\e)), u_k\mapsto\partial_x^k(u(x,{\bf s};\e)), k\geq0},
\end{align}
is a particular solution to the Toda lattice hierarchy~\eqref{todahier}. What is more, 
since $e^{\e^{-2}\F_0(x,{\bf s})}=e^{\e^{-2}\widetilde{\F}_0(x,{\bf s})}$ is the tau-function of 
the solution $(v(x,{\bf s}), u(x, {\bf s}))$ 
to the dispersionless Toda lattice hierarchy and using again 
the Theorem~1.1 of~\cite{DZ1}, we find that
\beq
\tau(x,{\bf s};\e):=\exp\left({\e^{-2}\widetilde{\F}_0(x,{\bf s})} + 
\sum_{g\geq1} \e^{2g-2}F_g^{\mathbb{P}^1}\big|_{v_k\mapsto\partial_x^k( v(x,{\bf s})), u_k\mapsto\partial_x^k(u(x,{\bf s})), k\geq0} \right)
\eeq
is the tau-function of the solution $\bigl(\widehat V(x,{\bf s};\e), \widehat U(x,{\bf s};\e)\bigr)$ to the Toda lattice hierarchy. 
Note that 
the functions $F_g^{\mathbb{P}^1}$, $g\geq 2$,  
satisfy the following equation:
\beq
 \frac{\p F_g^{\mathbb{P}^1}}{\p v} = 0,\label{stfg2022}
\eeq
which follows from the string equation~\eqref{stringp1} for the GW invariants of~$\mathbb{P}^1$.
By using \eqref{dilatongenuszero48}, \eqref{homogenfgp1}, \eqref{genuszerostringgue}, \eqref{stfg2022}
one can verify that this tau-function $\tau(x,{\bf s};\e)$ satisfies the following two relations:
\begin{align}
&\sum_{j\geq1}  \Bigl(s_j-\frac12 \delta_{j,2}\Bigr)  \frac{\p \tau(x,{\bf s};\e)}{\p s_j} + \e  \frac{\p \tau(x,{\bf s};\e)}{\p \e} + x  \frac{\p \tau(x,{\bf s};\e)}{\p x} 
+ \frac1{12}\tau(x,{\bf s};\e) = 0, \label{GUEdilaton325}\\
&\sum_{j\geq2} j \Bigl(s_j-\frac12 \delta_{j,2}\Bigr) \frac{\p \tau(x,{\bf s};\e)}{\p s_{j-1}} + \frac{x s_1}{\e^2} \tau(x,{\bf s};\e)  = 0, \label{GUEstring325}
\end{align}
which agree with the linear equations~\eqref{dilaton}, \eqref{string}.
The theorem is proved. 
\end{proof}

Several applications of Theorem~\ref{thm1} can be found in \cite{DLYZ2,DY1,DY2,DYZ}; some of the details are also given in the next section. 

\section{Topological gravity and matrix gravity}\label{section4}

In the previous sections, we studied the GUE partition function and give in Theorem~\ref{thm1}  
a jet representation for the genus~$g$ GUE free energy $\F_g(x,{\bf s})$ for $g\geq1$, obtained  
by the one for the genus~$g$ free energy of GW invariants of~$\PP^1$. In this 
section, we consider the restriction to {\it even couplings}, and 
revisit its connection to GW invariants of {\it a point} and the associated Hodge integrals.

\subsection{Identification in topological gravity}
In his seminal work~\cite{Witten}, Witten proposed two versions of two-dimensional quantum gravity: topological gravity 
and matrix gravity. In this subsection, let us consider the topological one, that is, following 
Witten~\cite{Witten}, the partition function of 
psi-class integrals on Deligne--Mumford's moduli space of curves~\cite{DM}. To be precise, 
let $\mathcal{F}_{\rm WK}({\bf t};\e)$, $g\geq0$, be the following generating series for psi-class 
integrals:
\beq
\mathcal{F}_{\rm WK}({\bf t};\e):=\sum_{g\geq0} \e^{2g-2} \sum_{k\geq 0} 
\sum_{i_1, \dots, i_k\geq 0} \frac{t_{i_1}\cdots t_{i_k}}{k!} \int_{\overline{\mathcal M}_{g,k}} 
\psi_1^{i_1}\cdots \psi_k^{i_k},
\eeq
called the {\it free energy}.
Here, ${\bf t}=(t_0,t_1,t_2,\dots)$ and $\e$ are indeterminates, $\overline{\mathcal M}_{g,k}$ denotes the 
moduli space of stable algebraic curves of genus~$g$ with $k$ distinct marked points, 
and $\psi_i$ ($1\le i\le k$) denotes the first Chern class of the $i$th cotangent line bundle on $\overline{\mathcal M}_{g,k}$.
Let 
\beq
\F_{\rm WK}({\bf t};\e) := \sum_{g\geq0} \e^{2g-2} \mathcal{F}^{\rm WK}_g({\bf t}).
\eeq
We call $\mathcal{F}^{\rm WK}_g({\bf t})$ {\it genus~$g$ part of the free energy} $\F_{\rm WK}({\bf t};\e)$. The exponential 
\beq\label{defZwk44}
\exp\bigl(\F_{\rm WK}({\bf t};\e)\bigr) =: Z_{\rm WK}({\bf t};\e) 
\eeq
is called the {\it partition function of psi-class integrals}. 
It was conjectured by Witten~\cite{Witten} and proved by Kontsevich~\cite{Kont} 
that the partition function 
$Z_{\rm WK}({\bf t};\e)$ is a particular tau-function for the Korteweg--de Vries (KdV) integrable hierarchy.
We also refer to $Z_{\rm WK}({\bf t};\e)$ as the {\it partition function for the topological quantum gravity}. 

Another important model regarding the intersection theory on~$\overline{\mathcal M}_{g,k}$ 
is the partition function of certain special cubic Hodge integrals \cite{DLYZ1, FP, LLZ, OP3}, 
which from its definition is a deformation of the partition function $Z_{\rm WK}({\bf t};\e)$ and 
has important relation to the GUE partition function \cite{DLYZ2, DY2}. 
To be precise,  define $Z_{\rm H} ({\bf t};\e)$ as follows:
\begin{align}
& Z_{\rm H} ({\bf t};\e) = e^{\mathcal{H}({\bf t};\e)}, \label{sh144}
\end{align}
where
\begin{align}
& \mathcal{H}({\bf t};\e):= \sum_{g\ge0} \e^{2g-2} \mathcal{H}_g({\bf t}), \label{sh244}\\
& \mathcal{H}_g({\bf t}):=\sum_{k\geq 0}
\sum_{i_1, \dots, i_k\geq 0} \frac{t_{i_1}\cdots t_{i_k}}{k!} \int_{\overline{\mathcal M}_{g,k}} 
\psi_1^{i_1}\cdots \psi_k^{i_k}\Lambda(-1)^2 \Lambda(\tfrac12),\quad g\geq0.  \label{sh344}
\end{align}
Here, $\Lambda(z):=\sum_{j=0}^g \lambda_j z^j$ is the Chern polynomial of the Hodge bundle $\mathbb{E}_{g,k}$ on $\overline{\mathcal M}_{g,k}$ with 
$\lambda_j$ being the $j$th Chern class of $\mathbb{E}_{g,k}$. We call $\mathcal{H}({\bf t};\e)$ the {\it Hodge 
free energy} and $Z_{\rm H} ({\bf t};\e)$ the {\it Hodge partition function}\footnote{The Hodge partition function considered 
in this paper is a specialization of the one in~\cite{DLYZ1}; geometric and topological significance of this specialization can be found 
e.g.~in~\cite{DLYZ1, DLYZ2, DY2}.}.
Being suggested by the {\it Hodge-GUE correspondence}~\cite{DLYZ2, DY2} (see also Theorem~\ref{HodgeGuethm45} below), we 
refer to the Hodge partition function 
$Z_{\rm H} ({\bf t};\e)$ defined in~\eqref{sh144}--\eqref{sh344}
as the {\it dual partition function for the topological quantum gravity}. 

In genus zero, we have the obvious equality 
\beq\label{genus0id44}
\mathcal{H}_0({\bf t})=\F^{\rm WK}_0({\bf t}).
\eeq
The discrepancy between the two partition functions 
$Z_{\rm H} ({\bf t};\e)$ and $Z_{\rm WK}({\bf t};\e)$
starts from their genus one parts. To understand this discrepancy, 
it will be convenient to look at their jet-representations \cite{DW, DY3, DZ-norm, EYY, Witten}. Recall the following lemma. 
\begin{lemma}  \label{lemma4428}
Denote 
\beq
v_{\rm WK}({\bf t}) := \frac{\p^2 \F^{\rm WK}_0({\bf t})}{\p t_0^2}.
\eeq
For each $g\geq 1$, there exist elements 
\begin{align}
& F^{\rm WK}_g(z_1,\dots,z_{3g-2})  \;\in\; \QQ\bigl[z_2,\dots,z_{3g-2}, z_1, z_1^{-1}\bigr], \\
& H_g(b_0,b_1,\dots,b_{3g-2}) \;\in\; \QQ\bigl[b_2,\dots,b_{3g-2}, b_0, b_1, b_1^{-1}\bigr],
\end{align}
such that 
\begin{align}
&\F^{\rm WK}_g({\bf t}) = F^{\rm WK}_g\biggl(\frac{\p v_{\rm WK}({\bf t})}{\p t_0},\dots, \frac{\p^{3g-2} v_{\rm WK}({\bf t})}{\p t_0^{3g-2}}\biggr), \\ 
&\mathcal{H}_g({\bf t}) = H_g\biggl(v_{\rm WK}({\bf t}),\frac{\p v_{\rm WK}({\bf t})}{\p t_0},\dots, \frac{\p^{3g-2} v_{\rm WK}({\bf t})}{\p t_0^{3g-2}}\biggr).
\end{align}
Moreover, for $g\geq2$, $H_g(b_0,b_1,\dots,b_{3g-2})$ does not depend on~$b_0$.  
\end{lemma}
\noindent See for example~\cite{DLYZ1, DY3} for the proof of this lemma. For the reader's convenience, we list the first few 
$F^{\rm WK}_g(z_1,\dots,z_{3g-2})$, $H_g(b_0,b_1,\dots,b_{3g-2})$ as follows:
\begin{align}
& F^{\rm WK}_1(z_1)=\frac1{24} \log z_1, \quad F^{\rm WK}_2(z_1,z_2,z_3,z_4) = \frac{z_4}{1152 z_1^2}-\frac{7 z_3 z_2}{1920 z_1^3}+\frac{z_2^3}{360 z_1^4}, \\
& H_1(b_0,b_1) =   \frac1{24} \log b_1  - \frac {b_0}{16},\\
& H_2(b_0,b_1,b_2,b_3,b_4) =  \frac{7 b_2}{2560}+\frac{11 b_2^2}{3840 b_1^2}-\frac{b_1^2}{11520}-\frac{b_3}{320 b_1} \\
& \qquad \qquad \qquad \qquad \quad +\frac{b_4}{1152 b_1^2}-\frac{7 b_3 b_2}{1920 b_1^3}+\frac{b_2^3}{360 b_1^4}. \nn
\end{align}
The elements $F^{\rm WK}_g(z_1,\dots,z_{3g-2})$ with $g\geq1$ can be 
calculated recursively by solving the DZ loop equation~\cite{DZ-norm}; 
 the elements $H_g(b_0,b_1,\dots,b_{3g-2})$ can also be calculated 
recursively by solving the DZ type loop equation~\cite{DLYZ2}, or, they can be calculated
by using the algorithm given in~\cite{DLYZ1}.

Introduce a gradation $\widetilde{\deg}$ in $\QQ[b_2,\dots,b_{3g-2}, b_0, b_1, b_1^{-1}]$ by assigning 
\beq
\widetilde{\deg} \, b_k=1, \quad \forall\, k\geq 0. 
\eeq
Then for $g\geq2$, $H_g(b_0,b_1,\dots,b_{3g-2})$ decomposes into the homogeneous parts with respect to $\widetilde{\deg}$ as 
follows:
\beq
H_g(b_0,b_1,\dots,b_{3g-2}) = \sum_{d=1-g}^{2g-2} H_g^{[d]}(b_0,b_1,\dots,b_{3g-2}),
\eeq
where $H_g^{[d]}(b_0,b_1,\dots,b_{3g-2})$ is homogeneous of degree~$d$ with respect to~$\widetilde{\deg}$.
We have (cf.~\cite{DLYZ1, DY3}) 
\beq\label{lowestdegreewk}
H_1(b_0,b_1)=F^{\rm WK}_1(b_1)-\frac1{16}b_0,\quad H_g^{[1-g]}(b_0,b_1,\dots,b_{3g-2}) = F^{\rm WK}_g(b_1,\dots,b_{3g-2}) ~(g\geq2).
\eeq
Namely, $H_g$ can be viewed as a specific deformation of $F^{\rm WK}_g$;
in the big phase space this is obvious (by definition), and we see the deformation 
in the jet space by equalities in~\eqref{lowestdegreewk}.

The following proposition says that under a coordinate transformation in the jet space, {\it remarkably},  
$H_g(b_0,b_1,\dots,b_{3g-2})$ becomes $F^{\rm WK}_g(z_1,\dots,z_{3g-2})$, $g\geq1$.
\begin{prop}\label{idgravity}
Under the transformation $B: (z_0,z_1,\dots) \rightarrow (b_0, b_1,\dots)$ (i.e., $b_i=B_i({\bf z})$, $i\geq0$),   
defined inductively from 
\beq\label{112414}
B_0({\bf z}) = - \log z_0, \quad \p''(z_0) = - \frac12 \frac{z_1}{\sqrt{z_0}}, \quad [\p', \p'']=0,
\eeq
we have the identities:
\beq\label{hgfwkg425}
4^{g-1} H_g\bigl(B_0({\bf z}),B_1({\bf z}),\dots,B_{3g-2}({\bf z})\bigr)=F^{\rm WK}_g(z_1,\dots,z_{3g-2}), \quad g\geq1.
\eeq
Here, $\p'$ is the derivation on $\QQ[z_0, z_1, z_1^{-1}, z_2,z_3,\dots]$ such that $\p'(z_i)=z_{i+1}$,
and $\p''$ is the derivation on $\QQ[b_0, b_1, b_1^{-1}, b_2,b_3,\dots]$ such that $\p''(b_i)=b_{i+1}$.
\end{prop}
The proof of Proposition~\ref{idgravity} using the Hodge-BGW correspondence is given in~\cite{YZQS}. 
An equivalent version of this proposition and the proof are given in~\cite{YZa}.
For the reader's convenience, let us list the first few terms of the change of jet-variables in Proposition~\ref{idgravity}:
\begin{align}
B_1({\bf z}) = \frac{z_1}{2z_0^{3/2}}, \quad B_2({\bf z})=\frac{z_1^2}{2z_0^3}-\frac{z_2}{4z_0^2}, \quad
B_3({\bf z}) = \frac18 \frac{z_3}{z_0^{5/2}} -\frac{15}{16} \frac{z_1 z_2}{z_0^{7/2}} + \frac{35}{32}\frac{z_1^3}{z_0^{9/2}}
\end{align}
with $B_0({\bf z})$ given already in~\eqref{112414}. This transformation is invertible, and let us list also the first 
few terms of the inverse transformation:
\begin{align}
& (B^{-1})_0({\bf b}) = e^{-b_0}, \quad (B^{-1})_1({\bf b}) = 2 e^{-\frac{3}{2} b_0} b_1, \quad (B^{-1})_2({\bf b}) =  e^{-2b_0} (-4 b_2 + 8 b_1^2), \\
& (B^{-1})_3({\bf b}) = e^{-\frac{5}{2} b_0} \bigl(8 b_3-60 b_2 b_1+50 b_1^3\bigr).
\end{align}
A closed formula for the map~$B^{-1}$ is found with Don Zagier~\cite{YZa}. 
We have the identity 
\beq\label{123428}
F^{\rm WK}_g\bigl((B^{-1})_1({\bf b}),\dots,(B^{-1})_{3g-2}({\bf b})\bigr) = 4^{g-1} H_g\bigl(b_0,b_1,\dots,b_{3g-2}\bigr), \quad g\geq1.
\eeq
In view of integrable systems, 
the relationship given in
Proposition~\ref{idgravity} reveals the space/time duality 
between the $q$-deformed KdV hierarchy (cf.~\cite{BCRR, Frenkel, LYZZ})
and the KdV hierarchy.

\subsection{Back to the matrix gravity}
In the previous subsection, we recalled the  
identification between the partition function~\eqref{defZwk44} and the dual partition function~\eqref{sh144} 
for the topological quantum gravity: 
for genus zero, it is given in the big phase space by~\eqref{genus0id44}; 
for higher genera, it is given in the jet-space by Proposition~\ref{idgravity}. 

In this subsection, following Witten~\cite{Witten}, we look at a certain reduction of the GUE partition function, 
which is regarded to as the matrix gravity. To be precise, define 
 the {\it even GUE partition function} $Z_{\rm even}(x,{\bf s}_{\rm even})$ by 
\beq\label{globaldefZeven}
Z_{\rm even}(x,{\bf s}_{\rm even}) := 
\frac{(2\pi)^{-n}\e^{-\frac1{12}}}{{\rm Vol}(n)} \int_{{\mathcal H}(n)} e^{-{\frac1\e}  \tr \, Q_{\rm even}(M; {\bf s}_{\rm even})} dM, \quad x=n\e,
\eeq
where ${\rm Vol}(n)$ is defined in~\eqref{voldef41}, and 
\beq
Q_{\rm even}(y; {\bf s}_{\rm even}) := \frac12 y^2 -\sum_{j \in \ZZ^{\rm even}_{\geq2}} s_{j} y^{j}.
\eeq 
Clearly, this partition function equals $Z(x,{\bf s})$ being restricted to ${\bf s}_{\rm odd}={\bf 0}$.

According to~\eqref{Fgue1x} and~\eqref{defcorrGUEfree}, the logarithm of $Z_{\rm even}(x,{\bf s}_{\rm even})$ has the expression
\begin{align}
& \log Z_{\rm even}(x,{\bf s}_{\rm even};\e) =:\F_{\rm even}(x,{\bf s}_{\rm even};\e) 
=: \sum_{g\geq0} \e^{2g-2} \F^{\rm even}_{g}(x,{\bf s}_{\rm even}) \\
& = 
\frac{x^2}{2\e^2}\biggl(\log x-\frac32\biggr) 
- \frac{\log x}{12} + \zeta'(-1) + \sum_{g\geq2} \frac{\e^{2g-2} B_{2g}}{4g(g-1)x^{2g-2}} \\
& \quad + \sum_{k\geq 1} \sum_{g\geq 0, \, j_1, \dots, j_k \in Z^{\rm even}_{\ge2} \atop 2-2g - k + |{\bf j}|/2\ge1} a_g({\bf j}) 
 s_{j_1} \cdots s_{j_k} \e^{2g-2} x^{2-2g - k + |{\bf j}|/2}.   \nn
\end{align}
We call $\F_{\rm even}(x,{\bf s}_{\rm even};\e)$ the {\it even GUE free energy}, and 
$\F^{\rm even}_{g}(x,{\bf s}_{\rm even})$ its genus~$g$ part.

The power series $u(x,{\bf s})$, $v(x,{\bf s})$ (cf.~\eqref{ELgue45}) being restricted to ${\bf s}_{\rm odd}={\bf 0}$, 
denoted by $u(x,{\bf s}_{\rm even})$, $v(x,{\bf s}_{\rm even})$, have the following explicit expressions~\cite{DY2}:
\begin{align}
& v(x,{\bf s}_{\rm even}) = 0, \label{veven43}\\
& e^{u(x,{\bf s}_{\rm even})} =  
\sum_{k=1} \frac1k \sum_{j_1,\dots,j_k\in \mathbb{Z}^{\rm even}_{\geq 0}, \atop j_1+\cdots+j_k=2k-2} {\rm wt} (j_1) \cdots {\rm wt}(j_k) 
\binom{j_1}{j_1/2} \cdots \binom{j_k}{j_k/2} s_{j_1} \cdots s_{j_k},\label{ueven43}
\end{align}
where we put $s_0=x$, and for $j\in\mathbb{Z}^{\rm even}_{\geq 0}$, 
\beq
{\rm wt} (j) := \left\{ \begin{array}{cc} 1, & j=0, \\ j/2, & {\rm otherwise}. \\ \end{array} \right.
\eeq

It is shown in~\cite{DY2} that one can take $v_0=v_1=v_2=\dots=0$ in 
$F^{\mathbb{P}^1}_g({\bf v}, {\bf v}_1,\dots,{\bf v}_{3g-2})$, $g\geq1$, yielding functions of $u_1, u_2,\dots$, denoted by 
$F^{\rm even}_g(u_1,\dots,u_{3g-2})$; explicitly, 
\beq\label{fevengdef}
F^{\rm even}_g(u_1,\dots,u_{3g-2}) := 
F^{\mathbb{P}^1}_g({\bf v}, {\bf v}_1,\dots,{\bf v}_{3g-2})\big|_{v_0=v_1=\dots=0} + \biggl(\zeta'(-1)-\frac{\log(-1)}{24}\biggr)\delta_{g,1}.
\eeq
For example, $F^{\rm even}_1=\frac1{12} \log u_1+\zeta'(-1)$. The expression for $F_g^{\rm even}$ with $g=2,\dots,5$ can be found in~\cite{DY2}.
The following theorem is then obtained.
\begin{theorem} [\cite{DY2}] \label{thm2}
The genus zero part of the even GUE free energy $\mathcal{F}_{{\rm even},0}(x,{\bf s}_{\rm even})$ has the expression:
\begin{align}
& \mathcal{F}^{\rm even}_0(x,{\bf s}_{\rm even}) = \frac12 x^2 u(x,{\bf s}_{\rm even}) 
+ x\sum_{j \in \mathbb{Z}_{\geq2}^{\rm even}} \binom{j}{j/2} \Bigl(s_{j}-\frac12\delta_{j,2}\Bigr) e^{\frac{j_1+j_2}2u(x,{\bf s}_{\rm even})} \label{F0even} \\ 
& \quad + 
\frac14\sum_{j_1,j_2 \in \mathbb{Z}_{\geq2}^{\rm even}} \frac{j_1j_2}{j_1+j_2} \binom{j_1}{j_1/2} \binom{j_2}{j_2/2} 
\Bigl(s_{j_1}-\frac12\delta_{j_1,2}\Bigr) \Bigl(s_{j_2}-\frac12\delta_{j_2,2}\Bigr)  e^{\frac{j_1+j_2}2u(x,{\bf s}_{\rm even})}, \nn
\end{align}
where $u(x,{\bf s}_{\rm even})=\log x + \cdots$ is given by~\eqref{ueven43}.
For $g\geq1$, the genus~$g$ part of the even GUE free energy 
$\mathcal{F}^{\rm even}_g(x,{\bf s}_{\rm even})$ satisfies that
\beq\label{fgfmgequal}
\mathcal{F}^{\rm even}_{g}(x,{\bf s}_{\rm even}) 
= F^{\rm even}_g\biggl(u(x,{\bf s}_{\rm even}), \frac{\p u(x,{\bf s}_{\rm even})}{\p x},\dots,\frac{\p^{3g-2} u(x,{\bf s}_{\rm even})}{\p x^{3g-2}}\biggr),
\eeq
where $F^{\rm even}_g(u,u_1,\dots,u_{3g-2})$ are defined by~\eqref{fevengdef}.
\end{theorem}

Let 
\beq
\Lambda=e^{\e \p_x}
\eeq
denote the shift operator. Following~\cite{DLYZ2} (cf.~also~\cite{DY2}), define
the {\it modified even GUE free energy} $\widetilde{\F}(x, {\bf s}_{\rm even};\e)$ by  
\beq\label{129415}
\widetilde{\F}(x, {\bf s}_{\rm even};\e) := \bigl(\Lambda^{1/2}+\Lambda^{-1/2}\bigr)^{-1} \bigl(\mathcal{F}_{{\rm even}}(x,{\bf s}_{\rm even};\e) \bigr) 
= : \sum_{g\geq0} \e^{2g-2} \widetilde{\F}_g(x, {\bf s}_{\rm even}).
\eeq
We call $\widetilde{\F}_g(x, {\bf s}_{\rm even})$ the {\it genus~$g$ modified GUE free energy}, and call 
the exponential $e^{\widetilde{\F}(x, {\bf s}_{\rm even})}=: \widetilde{Z}(x, {\bf s}_{\rm even})$ 
the {\it modified even GUE partition function}.

By definition~\eqref{129415} and by~\eqref{fgfmgequal} we see the followings~\cite{DY2}:
$\widetilde{\F}_0(x, {\bf s}_{\rm even}) = \F^{\rm even}_0(x, {\bf s}_{\rm even})/2$, and 
for $g\geq1$ $\widetilde{\F}_g(x, {\bf s}_{\rm even})$ 
admits the jet representation:
\beq
\widetilde{\F}_g(x, {\bf s}_{\rm even}) = 
\widetilde F_g \biggl(u(x,{\bf s}_{\rm even}), \frac{\p u(x,{\bf s}_{\rm even})}{\p x},\dots,\frac{\p^{3g-2} u(x,{\bf s}_{\rm even})}{\p x^{3g-2}}\biggr),
\eeq
where $\widetilde F_g(u,u_1,\dots,u_{3g-2})$, $g\geq1$, can be determined by 
\begin{align}
& \widetilde F_g(u,u_1,\dots,u_{3g-2}) \label{137428}\\
&= \frac{(-1)^g} 2 E_{2g} u_{2g-2} + \frac12 \sum_{g_1=1}^g (-1)^{g-g_1} E_{2g-2g_1} \p^{2g-2g_1} (F_{g_1}^{\rm even}(u,u_1,\dots,u_{3g_1-2})) \nn
\end{align}
with $E_k$ being the $k$th Euler number 
\[\p:=\sum_{k\geq0} u_{k+1} \p_{u_k}.\] 
Here the shuffling in genus phenomenon also appeared in~\cite{YZ}. We also have
\beq\label{138428}
F_g^{\rm even}(u_1,\dots,u_{3g-2}) = \frac{u_{2g-2}}{2^{2g} (2g)!} 
+ \sum_{m=1}^g \frac{2^{3m-2g}}{(2g-2m)!} \p^{2g-2m} \bigl(\widetilde F_m(u,u_1,\dots,u_{3m-2})\bigr),
\eeq
where $g\geq1$. 


The {\it Hodge-GUE correspondence}, conjectured in~\cite{DY2} (cf.~also~\cite{DLYZ1}) and proved in~\cite{DLYZ2}, is given by the following theorem.

\begin{theorem}[\cite{DLYZ2, DY2}] \label{HodgeGuethm45} The identity
\beq\label{Hodge-GUE49}
\widetilde Z(x,{\bf s}_{\rm even};\e) = \exp\biggl(\frac{A(x,{\bf s}_{\rm even})}{2\e^2}
+\frac{\zeta'(-1)}{2}\biggr) Z_{\rm H}\bigl({\bf t}(x,{\bf s}_{\rm even});\sqrt{2}\e\bigr),
\eeq
holds true in $\CC((\e^2))[[x-1,{\bf s}]]$, where 
\begin{align}
& A(x,{\bf s}_{\rm even})=\frac14 \sum_{j_1,j_2\in \ZZ^{\rm even}_{\ge 2}} \frac{j_1 j_2}{j_1+j_2} 
\binom{j_1}{j_1/2}\binom{j_2}{j_2/2} \Bigl(s_{j_1} - \frac{\delta_{j_1,2}}2\Bigr)\Bigl(s_{j_2} - \frac{\delta_{j_2,2}}2\Bigr) \label{defA48} \\
& \qquad \qquad \qquad + x \sum_{j\in \ZZ^{\rm even}_{\ge 2}} \binom{j}{j/2}\Bigl(s_j - \frac{\delta_{j,2}}2\Bigr), \nn
\end{align}
and 
\beq\label{tixs}
t_i(x,{\bf s}) = \sum_{j \in \ZZ^{\rm even}_{\ge 2}} (j/2)^{i+1} \binom{j}{j/2} \Bigl(s_j - \frac{\delta_{j,2}}2\Bigr) + \delta_{i,1} + x \delta_{i,0}, \quad i\geq0.
\eeq
\end{theorem}

Taking logarithms on both sides of the identity~\eqref{Hodge-GUE49}, we find that it is equivalent to 
\beq\label{Hodge-GUE492}
\widetilde \F(x,{\bf s}_{\rm even};\e) = \frac{A(x,{\bf s}_{\rm even})}{2\e^2}
+\frac{\zeta'(-1)}{2} + \mathcal{H}({\bf t}\bigl(x,{\bf s}_{\rm even});\sqrt{2}\e\bigr).
\eeq
The $g=0$ part of this identity is proved in~\cite{DY2}, and the higher genera parts are proved in~\cite{DLYZ2}. 
To understand more the higher genera parts of~\eqref{Hodge-GUE49}, again, we go to the jet space. 
The following lemma recalls the important relationship between $v_{\rm WK}({\bf t})$ and $u(x,{\bf s})$. 
\begin{lemma}[\cite{DY2}]
Under the substitution~\eqref{tixs}, the following identity is true:
\beq\label{uvkey}
v_{\rm WK}({\bf t}(x,{\bf s})) = u(x,{\bf s}).
\eeq
\end{lemma}
Using~\eqref{uvkey} and observing that 
\beq
\frac{\p}{\p t_0} = \frac{\p}{\p x},
\eeq
we can rewrite the higher genera parts of identity~\eqref{Hodge-GUE49} in the jet space as follows:
\beq\label{jetFHequal410}
\widetilde F_g (b_0,b_1,\dots,b_{3g-2}) = H_g(b_0,b_1,\dots,b_{3g-2}), \quad g\geq1.
\eeq
Therefore, we have identified the higher genera parts in jets of the modified even GUE partition function with those of the Hodge partition function. 
Then by using~\eqref{138428}, one comes back to the matrix gravity $F^{\rm even}_g$ in the higher genera 
with the topological gravity as a starting point, i.e.: $F_g^{\rm WK} \leftrightarrow H_g = \widetilde F_g \leftrightarrow F_g^{\rm even}$, $g\geq1$ 
(cf.~the diagram of the Introduction).



Comparing the lowest degree part of the identity~\eqref{jetFHequal410} with respect to~$\widetilde{\deg}$, 
using~\eqref{137428}, and noticing that the operator~$\p$ does not change~$\widetilde{\deg}$, 
we arrive at the following corollary, which 
explains the starting arrow on top of the square of the diagram of the Introduction.
\begin{cor} \label{cor1}
The following equalities are true:
\begin{align}
& 2 F^{\rm WK}_1(z_0,z_1) = F^{\rm even}_1(z_1) -\zeta'(-1) = \frac1{12} \log z_1, \label{146428}\\ 
& 2^g F^{\rm WK}_g(z_1,\dots,z_{3g-2}) = F_g^{{\rm even},[1-g]}(z_1,\dots,z_{3g-2}),\quad g\geq 2. \label{147428}
\end{align}
\end{cor}

\end{document}